\documentclass[12pt, a4paper,reqno]{amsart}

\usepackage[T1]{fontenc}
\usepackage[utf8]{inputenc}
\usepackage{xcolor}

\definecolor{bleu_sombre}{rgb}{0,0,0.6}  \definecolor{rouge_sombre}{rgb}{0.8,0,0}\definecolor{vert_sombre}{rgb}{0,0.6,0}
\usepackage[plainpages=false,colorlinks,linkcolor=bleu_sombre,
citecolor=rouge_sombre,urlcolor=vert_sombre,breaklinks]{hyperref}

\usepackage[english]{babel}
\usepackage{geometry}

\usepackage{amsmath,amssymb,amsthm,graphicx,amsfonts,enumerate,stmaryrd, mathrsfs,bbm}

\theoremstyle{plain}
\newtheorem{theorem}{{Theorem}}[section]
\newtheorem*{theorem*}{{Theorem}}
\newtheorem{proposition}[theorem]{Proposition}

\newtheorem*{proposition*}{Proposition}
\newtheorem{corollary}[theorem]{Corollary}
\newtheorem*{corollary*}{Corollary}
\newtheorem{lemma}[theorem]{Lemma}

\newtheorem*{lemma*}{Lemma}

\theoremstyle{definition}

\newtheorem*{definition*}{Definition}

\theoremstyle{remark}
\newtheorem{remark}[theorem]{Remark}

\makeatletter

\@addtoreset{equation}{section}
\makeatother

\renewcommand{\leq}{\leqslant}	\renewcommand{\geq}{\geqslant}

\newcommand{\R}{\mathbb{R}}

\newcommand{\dd}{\mathrm{d}}

\title[]{The Flea on the Magnetic Elephant}

\author[P.\ Exner]{Pavel Exner}
\address[P. Exner]{Doppler Institute for Mathematical Physics and Applied Mathematics\\
Czech Technical University\\
B\v rehov\'a 7\\ CZ-11519 Prague\\ and Department of
Theoretical Physics\\ NPI\\ Academy of Sciences\\ CZ-25068 \v{R}e\v{z}
near Prague, Czechia} \email{exner@ujf.cas.cz}
\urladdr{https://people.fjfi.cvut.cz/exnerpav/}

\author[L. Morin]{Léo Morin}
\address[L. Morin]{Department of Mathematics, University of Copenhagen, Universitetsparken 5, DK-2100 Copenhagen \O, Denmark}
\email{lpdm@math.ku.dk }

\begin{document}

	\maketitle

\begin{abstract}
We investigate a two-dimensional magnetic Laplacian with two radially symmetric magnetic wells. Its spectral properties are determined by the tunneling between them. If the tunneling is weak and the wells are mirror symmetric, the two lowest eigenfunctions are localized in both wells being distributed roughly equally. In this note we show that an exponentially small symmetry violation can in this situation have a dramatic effect, making each of the eigenfunctions localized dominantly in one well only. This is reminiscent of the `flea on the elephant' effect for Schrödinger operators; our result shows that it has a purely magnetic counterpart.
\end{abstract}	

\section{Introduction}

One of the remarkable effects in quantum mechanics concerns the situation when a tiny violation of symmetry of a double well system may cause a dramatic change in the associated bound state wave functions. Barry Simon \cite{Simon85} gave it the name `flea on the elephant' when he extended an earlier work done in \cite{JLMS81,GGJL84} to a multidimensional setting; in his words the flea does not change the shape of the elephant but it can irritate the elephant enough so that it shifts its weight.

The crucial element in understanding of this effect was an analysis of tunneling between symmetric wells, specifically the construction of an interaction matrix which effectively describes the bottom of the spectrum and the associated eigenfunctions. Tunneling in magnetic systems turned out to be more complicated, and for a long time only partial results were known \cite{HS87b}. This has changed by recent results which brought the first tunneling formul{\ae} with magnetic field under the radiality assumptions on the wells \cite{FSW22,FM25,FMR,HK22,HKS23,M24}; they involved new purely magnetic effects due to fast oscillations of the eigenfunctions. In \cite{GBFMR}, the first example of tunneling between non-radial purely magnetic wells was obtained.

The aim of the present note is to show that the `flea on the elephant' effect has a purely magnetic counterpart which shares with the original the mechanism based on the tunneling between two sufficiently separated wells. Imposing the  radiality assumptions on the wells, we adapt the techniques from \cite{FMR} in order to study perturbations of a purely magnetic symmetric double-well, and show that tiny perturbations can again have a dramatic effect on the eigenfunctions.

In the next section we formulate the problem and state our main result, Theorem~\ref{thm.doublewell}, reducing the semiclassical expansion at the bottom of the spectrum to a $2\times 2$ matrix problem; we also apply it to the case of our interest when the symmetry of double-well system is slightly perturbed. Section~\ref{sec.proof} is devoted to the proof of Theorem~\ref{thm.doublewell}.

\section{Perturbations of Symmetric Magnetic Wells}

\subsection{The problem setting and preliminaries}

As indicated, we will be concerned with purely magnetic double wells; our system is two-dimensional and exposed to magnetic field perpendicular to the plane. Specifically, we consider a pair of \emph{radial} magnetic wells centered at $x_\ell = \big(-\frac L2,0\big)$ and $x_r=\big(\frac L 2,0\big)$; the respective left and right fields are given by $\mathscr{C}^\infty$-smooth functions $B_\ell$ and $B_r$, such that
\begin{itemize}
\item[-] $B_\star >0$ is a radial function of $x-x_\star$,
\item[-] $B_\star$ has its unique minimum at $x_\star$ which is non-degenerate,
\item[-] $B_\star(x) = b_1$ holds for $|x-x_\star|\geq a$,
\end{itemize}
where $\star \in \lbrace \ell, r \rbrace$, and $a \in \big(0,\frac L2\big)$ and $b_1 >0$ are some constants. The double-well magnetic field $B$ is then defined as superposition of the two single wells,
\begin{equation}\label{def.B}
B(x) = \begin{cases}
B_\ell(x) &{\textrm{if}} \quad x_1 \le0,\\
B_r(x) &{\textrm{if}} \quad x_1 \ge0.
\end{cases}
\end{equation}
With each of the fields $B_\ell$ and $B_r$ we associated a vector potential satisfying $\nabla \times A_\star = B_\star$; we choose for them the circular gauge,
\[ A_\star(x) = \int_{[x_\star,x]} B_\star (y) (y-x_\star)^\perp \dd y, \quad y^\perp = (-y_2,y_1), \quad \star \in \lbrace \ell, r \rbrace.\]

Let us first recall some known facts about single-well system properties. The operators describing the individual wells are then given by
\[ \mathscr{L}_\star = (ih\nabla + A_\star)^2, \quad \star \in \lbrace \ell,r \rbrace. \]
Since each of the magnetic fields $B_\star$ have by assumption a unique and non-degenerate minimum, it is well known \cite{2Dradial,HKo,RV} that the bottom of the spectrum of $\mathscr{L}_\star$ is discrete and asymptotic expansion for the low eigenvalues can be obtained explicitly. In particular, according to \cite[Theorem~1.2]{2Dradial} we have

\begin{theorem} \label{thm.single.magnetic.well}
For $\star \in \lbrace \ell, r \rbrace$, the ground state of $\mathscr{L}_\star$ is a radial function of $x-x_\star$. Moreover, the eigenvalues at the bottom of the spectrum of $\mathscr{L}_\star$ satisfy
\[ \mu_{\star,j}(h) =h \min B_\ell + h^2 \big( 2j d_\star + d'_\star \big) + \mathscr{O}(h^3),\]
with
\[ d_\star = \frac{\sqrt{\det H_\star}}{\min B_\star}, \qquad d'_\star = \frac{\big( {\mathrm{tr}} H_\star^{\frac 12}\big)^2}{2\min B_\star}, \qquad H_\star = \frac 12\, {\rm{Hess}} B_\star(x_\star). \]
\end{theorem}

\smallskip

In the following we will use for simplicity the symbol $\mu_\star := \mu_{\star,1}$ for the ground state energy of $\mathscr{L}_\star$ and by $\Psi_\star$ we denote the corresponding ground state eigenfunction. Since it is a radial function, the eigenvalue equation becomes the 2D radial Schrödinger equation,
\[ -\frac{1}{r} \partial_r r \partial_r \Psi_\star + \frac{1}{r^2} \Big( \frac{1}{2\pi} \int_{D(x_\star,r)} B_\star \mathrm d x \Big)^2 = \mu_\star \Psi_\star ,\]
in the coordinate $r=|x-x_\star|$. In view of this observation, we can construct WKB approximations of $\Psi_\star$ in the spirit of \cite{H88,HS}. Note that such expansions can also be obtained in a magnetic non-radial setting \cite{WKB,2Dradial}, however, these are not known to be exponentially good approximations.

\begin{proposition}\label{prop.WKB}
For $\star \in \lbrace \ell, r \rbrace$, we put
\[d_\star(s):=\int_{0}^{s}\frac{1}{2 \pi v^2} \left(\int_{D(x_\star,v)} B_\star \mathrm{d}x\right) v \mathrm{d}v\,.\]	
There exists a sequence of smooth functions $(a_{j})_{j\geq 0}$ on $\R^2$, radial with respect to $x-x_\star$, with $a_{0}>0$ such that the following holds: for any given $p\in\mathbb{N}$, setting $a_{h,\star}(x)=\sum_{j=0}^p h^j a_j(x)$, and considering
\[\Psi^{\mathrm{WKB}}_\star(x)=h^{-\frac12}a_{h,\star}(x)e^{-d_\star(|x-x_\star|)/h}\,,\]
we have
\[e^{d_\star(|x-x_\star|)/h}\left(\mathscr{L}_{\star}-\mu_\star\right)\Psi_\star^{\mathrm{WKB}}=\mathscr{O}(h^{p+1})\,.\]
Moreover, for all $\epsilon\in(0,1)$, there exist positive $C, h_0>0$ such that, for $h\in(0,h_0)$,
\begin{equation}\label{eq.Agmon}
\|(-ih\nabla-{A}_\star)(e^{(1-\epsilon)d_\star}\Psi_\star)\|\leq Ch\|\Psi_\star\|
\end{equation}
holds for all $h\in(0,h_0)$, and we also have the approximation
\[e^{d_\star(|x-x_\star|)/h}(\Psi_{\star}(x)-\Psi_\star^{\mathrm{WKB}}(x))=\mathscr{O}(h^{p+1})\,.\]
\end{proposition}
\noindent In particular, the inequality \eqref{eq.Agmon} provides us with exponential decay estimates on the eigenfunction $\Psi_\star$.

\subsection{Magnetic double wells}

Let us now turn to our proper topic. For the double-well magnetic field \eqref{def.B} we choose a vector potential $A$ such that $\nabla \times A = B$ and the system dynamics is governed by the self-adjoint realization of
\[ \mathscr{L}=(ih\nabla+A)^2 \;\; \text{ on } \; \mathbb R^2\,;\]
in contrast to the previous subsection we do not need to fix a particular gauge.

The central object here is the subspace associated with the spectral projection $\Pi = \mathbf{1}_{(-\infty,\Lambda ]}\big( \mathscr{L} \big)$, where $\Lambda := \min \big( \frac{\mu_{\ell}+\mu_{\ell,2}}{2}, \frac{\mu_r + \mu_{r,2}}{2} \big) $. If the two ground state energies are close to each other, one can construct a basis of ${\rm{Ran}} \, \Pi$ using the corresponding left and right well eigenfunctions, $\Psi_\ell$ and $\Psi_r$.

\begin{theorem}\label{thm.doublewell}
Assume that $\mu_\ell - \mu_r = o(h^2)$ and $L> (2+\sqrt{6})a$. Then there exists an orthonormal basis $(e_\ell,e_r)$ of ${\rm{Ran}} \, \Pi$ such that
\[ e_\ell = \Psi_\ell + \mathscr{O}(h^\infty), \quad e_r = \Psi_r + \mathscr{O}(h^\infty),\]
and the matrix representing $\mathscr{L}$ in this basis is
\[ M =
\begin{pmatrix}
\mu_\ell & w \\
\overline{w} & \mu_r
\end{pmatrix}
+ o(w) \]
with $w= c_0 h^{\nu} e^{-\frac{S}{h}} (1+o(1))$ as $h\to 0$ for some constants $c_0 \neq 0$ and $\nu \in \mathbb R$. Here
\begin{equation}\label{defS}
 S= d_\ell \left(\textstyle{\frac L2}\right) + d_r \left(\textstyle{\frac L2} \right) + I,
\end{equation}
where the $d_\star$ are defined in Theorem~\ref{thm.single.magnetic.well}, $I$ is a positive number given by
\[
I= \frac{M}{2} - \frac{b_1L^2}{8} + \frac{b_1L}{2} \sqrt{\frac{L^2}{4}-\frac{M}{b_1}} - M \ln \left( 1 + \sqrt{ 1 - \frac{4M}{b_1 L^2}}\right),
\]
and $M = \frac{1}{2\pi} \int_{\R^2} (2 b_1 - B_\ell - B_r)\,\mathrm{d}x$.
\end{theorem}

\smallskip

We postpone the proof to Section \ref{sec.proof} below and focus first on the consequences of the theorem. It is easy to diagonalize the $2\times 2$ matrix $M$; as a corollary we deduce asymptotic behavior of the eigenfunctions of $\mathscr{L}$. The instability of these eigenfunctions  we are interested in is then a consequence of the instability of the eigenvectors of $M$.

\begin{corollary}\label{coro1}
Assume $\mu_\ell - \mu_r = o(h^2)$ and $L>(2+\sqrt{6})a$.
\begin{enumerate}
\item If $|\mu_\ell - \mu_r| \ll w$, then the normalized eigenfunctions $\Psi_1$ and $\Psi_2$ of $\mathscr{L}$ are localized in both wells roughly equally,
\[ \int_{x_1 >0} |\Psi_j(x)|^2 \dd x = \frac 12 + o(1),\quad j=1,2,\]
as $h \to 0$.
\item If $|\mu_\ell - \mu_r| \gg w$ with $\mu_\ell < \mu_r$, then the normalized groundstate $\Psi_1$ is localized in the left well, and the first excited state $\Psi_2$ in the right well,
\[ \int_{x_1>0} |\Psi_1(x) |^2 \dd x =  o(1), \quad \int_{x_1>0} |\Psi_2(x)|^2 \dd x = 1+o(1),\]
as $h \to 0$.
\end{enumerate}
\end{corollary}

\begin{remark}
By symmetry, when $|\mu_\ell - \mu_r | \gg w$ with $\mu_\ell > \mu_r$, the ground state will be localized in the right well instead. The ground state is localized in the deeper well with the smaller ground-state energy, and the first excited state is correspondingly localized in the other one.
\end{remark}
\begin{proof}
The matrix $M$ introduced in Theorem \ref{thm.doublewell} has the eigenvalues
\[ \lambda_\pm =  \frac{\mu_\ell + \mu_r}{2} \pm \frac 12 \sqrt{(\mu_\ell - \mu_r)^2 + 4 |w|^2} +  o(w)\]
with the positive spectral gap $\lambda_+ - \lambda_- = \sqrt{(\mu_\ell - \mu_r)^2 + 4 |w|^2} + o(|w|) >0$. The corresponding eigenvectors satisfy
\[ v_\pm = \left( \frac{\mu_\ell - \mu_r}{2|w|} \mp \frac{\sqrt{(\mu_\ell - \mu_r)^2 + 4 |w|^2}}{2 |w|}, 1 \right)^\texttt{T} + o(\| v_\pm \|).\]
Consequently, in case $(1)$ we get
\[
\lambda_{\pm} = \frac{\mu_\ell + \mu_r}{2} \pm |w| + o(|w|),
\]
and the normalized eigenvectors are
\[ \frac{v_\pm}{\| v_\pm \|} = \frac{1}{\sqrt 2} {\mp 1 \choose 1} + o(1). \]
It means that the first two eigenfunctions of $\mathscr{L}$ are of the form
\[ \Psi_\pm = \frac{1}{\sqrt{2}}( e_\ell \mp e_r ) + o(1) = \frac{1}{\sqrt{2}}( \Psi_\ell \mp \Psi_r ) + o(1), \]
and the result follows since $\Psi_\ell$ and $\Psi_r$ are exponentially localized near $x=x_\ell$ and $x=x_r$, respectively.

In contrast, case (2) offers a different picture, the eigenvalues and eigenvectors being
\[ \lambda_- = \mu_\ell + \mathscr{O} \Big( \frac{|w|^2}{|\mu_\ell - \mu_r|} \Big), \quad \lambda_+ = \mu_r + \mathscr{O} \Big( \frac{|w|^2}{|\mu_\ell - \mu_r|} \Big),\]
and
\[ \frac{v_-}{\| v_- \|} = {1\choose 0} + o(1), \quad  \frac{v_+}{\| v_+ \|} = {0\choose 1} + o(1).\]
Hence the two first eigenfunctions of $\mathscr{L}$ are
\[ \Psi_1 = e_\ell + o(1) = \Psi_\ell + o(1), \quad \Psi_2 = e_r+ o(1) = \Psi_r + o(1),\]
and the result follows again.
\end{proof}

\subsection{Application: a magnetic flea.} Let us now turn to the effect indicated in the introduction and consider a slight perturbation of one of the wells in the originally symmetric pair -- for definiteness we choose the right one -- and inspect the effect on the eigenfunctions. For simplicity we restrict ourself to the situation when the perturbation is radially symmetric and at the same time supported away from the well center -- our flea is thus weak but not local its support being an annular set -- although \emph{we conjecture that the effect will persist in the absence of the radial symmetry hypothesis}. We thus consider magnetic fields
\[ B_\ell(x) = B^0(x-x_\ell), \quad B_r(x) = B^0(x-x_r) + t \beta(x-x_r),\]
where $t$ is a small parameter, and $B^0$, $\beta$ are radial functions satisfying the above assumptions.
\begin{theorem}\label{thm.mag.flea}
Assume that $\beta\geq 0$ is supported in $D(0,a) \setminus \lbrace 0 \rbrace$ and denote
\[R= \inf \lbrace |x|, x\in {\rm supp} \, \beta \rbrace >0.\]
Then for any $\varepsilon >0$ and $h$ small enough we have
\[ t\, e^{- (1+\varepsilon) \frac{d_\ell (R)}{h}} \leq | \mu_\ell - \mu_r | \leq t\, e^{-(1-\varepsilon) \frac{d_\ell(R)}{h}} .\]
\end{theorem}

\smallskip

The flea-on-the-elephant effect then follows from the last theorem in combination with Corollary \ref{coro1}. We define the quantity
\[
I_0 = \frac{M}{2} - \frac{b_1L^2}{8} + \frac{b_1L}{2} \sqrt{\frac{L^2}{4}-\frac{M}{b_1}} - M \ln \left( 1 + \sqrt{ 1 - \frac{4M}{b_1 L^2}}\right),
\]
where $L$ is the distance of the well centers, $b_1$ is the `background' magnetic field intensity, and $M = \frac{1}{2\pi} \int_{\R^2} (2b_1 - B^0(x-x_\ell) - B^0(x-x_r))\, \dd x$. Obviously, $I_0$ is the leading term of the quantity $I$ from Theorem \ref{thm.doublewell}, $I = I_0 + \mathscr{O}(t)$.

\begin{corollary} \label{cor.mgflea}
Let $L> (2+\sqrt 6)a$, then under the assumption of Theorem \ref{thm.mag.flea} we have the following dichotomy:
\begin{enumerate}
\item If $t \leq e^{- (1+\varepsilon)\frac{ 2 d_\ell (L/2) - d_\ell(R) + I_0}{h}}$ holds for some $\varepsilon >0$, then the normalized eigen\-functions $\Psi_1$ and $\Psi_2$ of $\mathscr{L}$ are localized in both wells roughly equally,
\[ \int_{x_1 >0} |\Psi_j(x)|^2 \dd x = \frac 12 + o(1),\quad j=1,2,\]
as $h \to 0$.
\item In contrast, if $h \gg t \geq e^{- (1-\varepsilon)\frac{ 2 d_\ell (L/2) - d_\ell(R) + I_0}{h}}$ holds for some $\varepsilon >0$, then each of the normalized eigenfunctions $\Psi_1$ and $\Psi_2$ of $\mathscr{L}$ is localized dominantly in a single well,
\[ \int_{x_1>0} |\Psi_1(x) |^2 \dd x = o(1), \qquad \int_{x_1>0} |\Psi_2(x)|^2 \dd x = 1+ o(1),\]
as $h \to 0$.
\end{enumerate}
\end{corollary}
In particular, since the right-hand side of the inequality in case (2) approaches zero faster than the left-hand one as $h \to 0$, even a very small perturbations of a symmetric double-well magnetic field can dramatically change the behavior of the eigenfunctions at the bottom of the spectrum.
\begin{proof}[Proof of Theorem \ref{thm.mag.flea}]
Since the perturbation is radially symmetric, by centering the single-well problems we see that $\mu_\ell=\mu_0$ and $\mu_r=\mu_t$ are the ground state energies of
\[ \mathscr{M}_0 = (ih\nabla + A^0)^2, \quad {\rm and} \quad \mathscr{M}_t = (ih \nabla + A^0 + t A^1)^2 \]
respectively, where $A^0(x) = \int_{[0,x]} B^0(y)y^\perp \,\dd y$ and $A^1(x) = \int_{[0,x]} \beta(y) y^\perp \,\dd y$. Moreover, by Theorem \ref{thm.single.magnetic.well} the ground-state eigenfunctions are independent of angle, hence $\mu_\ell$ and $\mu_r$ are the ground state energies of the following radial operators,
\[\mathscr P_0 = - \frac{h^2}{r} \partial_r r \partial_r + \frac{a(r)^2}{r^2}, \quad \mathscr{P}_t =  - \frac{h^2}{r} \partial_r r \partial_r + \frac{(a(r) + t \alpha(r))^2}{r^2},\]
respectively, where $a(r) :=\int_0^r B^0(s)s \dd s$ and $\alpha(r) := \int_0^r \beta(s)s\dd s$. We denote the associated normalized eigenfunctions by $\psi_0$ and $\psi_t$, i.e.
\[ \mathscr{P}_0 \psi_0 = \mu_\ell \psi_0, \quad \mathscr{P}_t \psi_t = \mu_r \psi_t. \]
Note that the perturbation has exponentially small effect,
\[ (\mathscr{P}_t - \mu_0) \psi_0 = \frac{2t \alpha(r)a(r) + t^2 \alpha(r)^2}{r^2} \psi_0 = \mathscr{O}(t\,e^{- \frac{d_\ell(R)}{h}}),\]
since $\alpha$ is by assumption supported on $[R,\infty)$ and the function $\psi_0$ decays by Theorem~\ref{thm.single.magnetic.well} like $\exp(- d_\ell(r)/h)$; note that $d_r>d_\ell$. Using the min-max principle we get $\inf\mathrm{spec}(\mathscr{P}_t-\mu_0)\leq \langle\psi_0, (\mathscr{P}_t-\mu_0) \psi_0\rangle$ so that
\[ | \mu_r - \mu_\ell | = \mathscr{O}\big(t\,e^{- \frac{d_\ell(R)}{h}}\big) \]
holds as $h\to 0$. Since the perturbation is bounded and the ground-state eigenvalue of a single well system is simple, we also have
\[ \| \psi_t - \psi_0 \| = \mathscr{O}\big(t\,e^{- \frac{d_\ell(R)}{h}}\big),\]
which verifies the second inequality of Theorem~\ref{thm.mag.flea}.

To get a lower bound on the eigenvalue gap, we use a perturbation argument again. We write $\mathscr{P}_t = \mathscr{P}_0 + \mathscr{R}$, $\delta \mu = \mu_r - \mu_\ell$ and $\delta \psi = \psi_t - \psi_0$, then the eigenvalue equation reads
\[( \mathscr{P}_0 + \mathscr{R})(\psi_0 + \delta \psi) = (\mu_0 + \delta \mu) (\psi_0 + \delta \psi),\]
which implies
\[ (\mathscr{P}_0 - \mu_0) \delta \psi + \mathscr{R} (\psi_0 + \delta \psi) = \delta \mu (\psi_0 + \delta \psi). \]
Next we take the scalar product with $\psi_0$ to find that
\[ \delta \mu ( 1 + \langle \delta \psi, \psi_0 \rangle) = \langle \mathscr{R}(\psi_0 + \delta \psi), \psi_0 \rangle .\]
Furthermore, we have $\langle \delta \psi, \psi_0 \rangle = \mathscr{O}\big(e^{- \frac{d_\ell(R)}{h}}\big)$ and
\[ | \langle \mathscr R \delta \psi, \psi_0 \rangle | \leq \| \delta \psi \| \| \mathscr{R} \psi_0 \| = \mathscr{O}(e^{- \frac{2d_\ell(R)}{h}}),\]
where we have used Schwarz inequality together with the support properties of $\mathscr{R}$ and the decay of $\psi_0$ again. This yields
\begin{equation}\label{eq.deltamu}
 \delta \mu = \langle \mathscr{R} \psi_0, \psi_0 \rangle \big( 1 + \mathscr{O}(e^{- \frac{d_\ell(R)}{h}}) \big) +\mathscr{O}(e^{- \frac{2d_\ell(R)}{h}}).
\end{equation}
The leading term can be written explicitly as
\[ \langle \mathscr{R} \psi_0, \psi_0 \rangle = \int_{R}^a \Big( \frac{2t\alpha(r)a(r) + t^2\alpha(r)^2}{r^2} \Big) |\psi_0(r)|^2 r \dd r > 0,\]
which means that for any $\varepsilon >0$ we have the lower bound
\[ \langle \mathscr{R} \psi_0, \psi_0 \rangle \geq \int_{R+\varepsilon}^a \Big( \frac{2t\alpha(r)a(r) + t^2\alpha(r)^2}{r^2} \Big) |\psi_0(r)|^2 r \dd r.\]
Now we use the Laplace method, the exponential decay of $\psi_0$ and the fact that the integrated function is positive near $R+\varepsilon$, obtaining
\[ \langle \mathscr{R} \psi_0, \psi_0 \rangle \geq C_\varepsilon t e^{-\frac{d(R+\varepsilon)(1 + \varepsilon)}{h}}, \]
where the additional $(1+\varepsilon)$ was used to absorb the powers of $h^{-1}$ coming from the WKB approximation of Proposition \ref{prop.WKB}. Using \eqref{eq.deltamu}, we deduce that
\[ |\delta \mu | \geq C_\varepsilon t e^{- \frac{d(R+\varepsilon)(1+\varepsilon)}{h}},\]
and the result follows since $\varepsilon$ was arbitrary.
\end{proof}

\begin{proof}[Proof of Corollary~\ref{cor.mgflea}]
In view of Corollary \ref{coro1} it is sufficient to compare the gap $|\mu_r - \mu_\ell|$ to the quantity $w$ of Theorem \ref{thm.doublewell}; the $\varepsilon$ in the exponent can again absorb in the limit $h\to 0$ the power of $h$ which $w$ may contain.
\end{proof}

\section{Proof of Theorem \ref{thm.doublewell}}
\label{sec.proof}

Let us now return to the result on which all the previous discussion depends. The proof scheme is adapted from \cite{FMR}, where the case of mirror-symmetric wells was considered.

\subsection{Construction of quasimodes}

To change the  gauge from $A$ to $A_\ell$ or $A_r$, we introduce the functions $\sigma_\ell$ and $\sigma_r$ defined by
\begin{equation}\label{def.sigma}
 \sigma_\star (x) = \int_{[0,x]} (A-A_\star) \cdot {\rm d}s,
 \end{equation}
where $[0,x]$ is the segment connecting the point $x\in\R^2$ with the origin of coordinates and the dot means scalar product. Note that they satisfy
\begin{equation}
\begin{cases}
\nabla \sigma_\ell = A - A_\ell  &{\rm for} \quad x_1 < \frac L2 - a,\\
\nabla \sigma_r = A-A_r &{\rm for} \quad x_1 > - \frac L2 + a.
\end{cases}
\end{equation}
In particular, for all $\Psi$ we have
\begin{equation}\label{eq.gaugechange}
 \mathscr{L}( e^{i \frac{\sigma_\ell}{h}} \Psi) = e^{i\frac{\sigma_\ell}{h}} \mathscr{L}_\ell \Psi, \quad {\rm on} \quad x_1 < \frac L2 - a,
 \end{equation}
and a similar result holds also for $\mathscr{L}_r$.

Using rough decay estimates on the eigenfunctions, we easily find that the low-lying spectrum of the double-well operator $\mathscr{L}$ is situated close to $\lbrace \mu_{\ell,1}, \mu_{r,1} \rbrace$.
\begin{lemma}\label{lem.roughestimates}
Assume that $\lambda \in \mathrm{spec}( \mathscr{L})$ is such that for all $h$ small enough the inequalities $|\lambda - \mu_{\ell,1}| < |\lambda - \mu_{\ell,2}|$ and $|\lambda - \mu_{r,1}| < |\lambda - \mu_{r,2}|$ hold. Then
\[ \min \big( |\lambda - \mu_{\ell,1} |, |\lambda - \mu_{r,1}| \big) = \mathscr{O}(h^\infty).\]
\end{lemma}

\smallskip

\noindent Note that the condition on $\lambda$ is meaningful when $\mu_{\ell,1}$ and $\mu_{r,1}$ are close to each other (up to $o(h^2)$). This condition makes sense since our aim is to study cases when the right well is a perturbation of the left one. In particular, it is satisfied for $\lambda < \Lambda$; recall that $\Lambda$ is the minimum of the midpoints of the gaps between the two first eigenvalues left and right.
\begin{proof}
Let $\Psi$ be an eigenfunction of $\mathscr{L}$ with eigenvalue $\lambda$. We use two cutoff functions $\theta_\ell$ and $\theta_r$, smooth so that $\theta_\star\Psi$ belongs to the domain of $\mathscr{L}_\star$, and such that $\theta_\ell^2 + \theta_r^2 =1$ and $\theta_\star$ equals one on a neighborhood of the well $D(x_\star, a)$. Using spectral theorem in combination with the fact that the closest eigenvalue of $\mathscr{L}_\star$ to $\lambda$ is $\mu_{\star,1}$, we get the estimate
\begin{align*}
 \| (\mathscr{L}_\ell - \lambda) e^{-i \frac{\sigma_\ell}{h}} \theta_\ell \Psi \|^2 +  \| (\mathscr{L}_r - \lambda) e^{-i \frac{\sigma_r}{h}} \theta_r \Psi \|^2 &\geq  |\lambda - \mu_{\ell,1}| \| \theta_\ell \Psi \|^2 + |\lambda - \mu_{r,1}| \| \theta_r \Psi \|^2  \\
 &\geq \min \big( |\lambda - \mu_{\ell,1}|, |\lambda - \mu_{r,1} | \big) \|\Psi \|^2.
 \end{align*}
Using further relation \eqref{eq.gaugechange} which holds on the support of $\theta_\ell$ and its right analogue, we arrive at
\begin{equation}\label{eq.1629}
\| (\mathscr{L}-\lambda) \theta_\ell \Psi \|^2 + \| (\mathscr{L}-\lambda) \theta_r \Psi \|^2 \geq \min \big( |\lambda - \mu_{\ell,1}|, |\lambda - \mu_{r,1} | \big) \|\Psi \|^2.
\end{equation}
Since $\lambda$ is by assumption an eigenvalue of $\mathscr{L}$, the two expressions on the left-hand side can be be bounded from above using the commutator,
\[ \| (\mathscr{L} - \lambda) \theta_\star \Psi \|^2 = \| \big[ \mathscr{L}, \theta_\star \big] \Psi \|^2.\]
The Agmon estimate (see, for instance, \cite[Sect.~1.4]{GBRVN21}) tells us that $\Psi$ decays exponentially away from the wells, which is where the commutators are supported. Therefore,
\[ \| (\mathscr{L} - \lambda) \theta_\star \Psi \|^2 = \mathscr{O}(h^\infty)\| \Psi \|^2,\]
which in combination with \eqref{eq.1629} yields the sought estimate.
\end{proof}

In order to get better estimates on the eigenvalues and eigenfunctions, we construct next an appropriate eigenbasis starting from the left and right well ground states, $\Psi_\ell$ and $\Psi_r$, respectively. We use again smooth enough cut-off functions; this time they will be $\chi_\ell$, $\chi_r$ such that, for a sufficiently small $\eta >0$ we have
\[ \chi_\ell(x_1) = \begin{cases}
1 &{\rm{if}} \quad x_1 < \frac L 2 - a - 2 \eta, \\
0 &{\rm{if}} \quad x_1 > \frac L2 - a-\eta,
\end{cases}
\]
and $\chi_r(x_1) = \chi_\ell(-x_1)$.Then we define
\begin{equation}\label{def.psil}
 \psi_\ell = e^{i \frac{\sigma_\ell}{h}}\chi_\ell \Psi_\ell, \qquad \psi_r = e^{i \frac{\sigma_r}{h}} \chi_r \Psi_r,
 \end{equation}
and furthermore, using the spectral projection $\Pi = \mathbf{1}_{(-\infty,\Lambda ]}\big( \mathscr{L} \big)$, we put
\[ g_\ell = \Pi \psi_\ell, \qquad g_r = \Pi \psi_r.\]
Then we claim that $\psi_\ell$ and $\psi_r$ are exponentially good quasimodes for $\mathscr{L}$, and $g_\ell$, $g_r$ are small perturbations of them.
\begin{lemma}\label{lem.g-psi}
For $\star \in \lbrace \ell, r \rbrace$, we have
\[ \| (\mathscr{L}-\mu_\star) \psi_\star \| = \mathscr{O}\left( e^{-d_\star \left(L - a - 3\eta\right)/h} \right), \quad \| \psi_\star \| = 1 + \mathscr{O}\left( e^{-d_\star \left(L - a - 3\eta\right)/h} \right).\]
Moreover, if $|\mu_\ell - \mu_r| = o(h^2)$, we have
\[ \langle \mathscr{L} (g_\star - \psi_\star), g_\star - \psi_\star \rangle =  \mathscr{O}\left( e^{-2 d_\star \left(L - a - 4\eta\right)/h} \right), \quad
\| g_\star - \psi_\star \| =  \mathscr{O}\left( e^{-d_\star \left(L - a - 4\eta\right)/h} \right). \]
\end{lemma}
\begin{proof}
The proof follows the same lines as \cite[Lemma 4.3]{FMR}. Using \eqref{eq.gaugechange} which holds on the support of $\chi_\ell$ we have
\[ \| (\mathscr L - \mu_\ell) \psi_\ell \| = \| (\mathscr{L}_\ell - \mu_\ell) \chi_\ell \Psi_\ell \| = \| \big[ \mathscr L_\ell, \chi_\ell \big] \Psi_\ell \| = \mathscr{O}\left( e^{-d_\ell \left(L - a - 3\eta\right)/h} \right),\]
where in the last estimate we used the exponential decay of $\Psi_\ell$ (Proposition \ref{prop.WKB}) and the fact that the commutator is nonzero (and bounded) only where $\chi_\ell$ is non-constant. Taking into account that $\Psi_\ell$ is normalized, we get in a similar way the estimate on $\| \psi_\ell \|$, and the analogous inequalities hold for $\psi_r$.

The estimates on $g_\star - \psi_\star$ can be obtained in the following way: by definition of the spectral projection $\Pi$ we have
\[  \langle (\mathscr{L}-\mu_\star) (I-\Pi) \psi_\star, (I-\Pi) \psi_\star \rangle \geq (\Lambda - \mu_\star ) \| (I-\Pi) \psi_\star \|^2.\]
Moreover, $\mathscr{L}$ commutes with its spectral projection $\Pi$, thus
\[ \langle (\mathscr{L}-\mu_\star) (I-\Pi) \psi_\star, (I-\Pi) \psi_\star \rangle \leq \| (\mathscr{L}-\mu_\star) \psi_\star \| \| (I-\Pi) \psi_\star \|.\]
Combining the last two inequalities, we infer that
\[ \| (I-\Pi) \psi_\star \| \leq \big(\Lambda-\mu_\star \big)^{-1}\| (\mathscr{L}-\mu_\star) \psi_\star \|.\]
Since $\mu_\ell = \mu_r + o(h^2)$ and $\mu_{\ell,2} - \mu_\ell \geq c h^2$ for some $c>0$ by Theorem \ref{thm.single.magnetic.well}, the gap between $\mathscr{L}$ and $\mu_\star$ is of order $h^{-2}$ which means that
\[ \| (I-\Pi) \psi_\star \| = \mathscr{O}\left( h^{-2} e^{-d_\star \left(L - a - 3\eta\right)/h} \right). \]
Furthermore, we can absorb the power factor in the exponential one changing $3\eta$ to $4 \eta$. Finally, to get the quadratic form bound, it is sufficient to use the Cauchy-Schwarz inequality,
\[ \langle \mathscr{L} (I-\Pi) \psi_\star, (I-\Pi) \psi_\star \rangle \leq \| \mathscr{L} (I-\Pi) \psi_\star \| \| (I-\Pi) \psi_\star \| = \mathscr{O}\left( e^{-2 d_\star \left(L - a - 4\eta\right)/h} \right), \]
where the last estimate follows from $ \| (\mathscr{L}-\mu_\star) \psi_\star \| = \mathscr{O}\left( e^{-d_\star \left(L - a - 3\eta\right)/h} \right)$.
\end{proof}

In fact, the pair $(g_\ell, g_r)$ spans the spectral subspace associated with $\Pi$.

\begin{lemma}
Assume that $\mu_\ell - \mu_r = o(h^2)$. Then for all $h$ small enough, $(g_\ell,g_r)$ is a basis of the subspace ${\rm{Ran}}\, \Pi$.
\end{lemma}
\begin{proof}
By Lemma \ref{lem.g-psi}, $(g_\ell, g_r)$ is a perturbation of $(\psi_\ell,\psi_r)$, and therefore $g_\ell$ and $g_r$ are linearly independent for $h$ small enough. To prove that they form a basis of ${\rm{Ran}} \, \Pi$, it suffices to check that the only $\Psi \in {\rm{Ran}}\, \Pi$ orthogonal to both $g_\ell$ and $g_r$ is zero; note that by definition of vectors $g_\star$ this orthogonality requirement implies
\begin{equation}\label{eq.orthogonal}
 \langle \Psi, \psi_\star \rangle = \langle \Pi \Psi, \psi_\star \rangle = \langle \Psi, g_\star \rangle = 0.
 \end{equation}
We use the same cutoff functions $\theta_\ell$ and $\theta_r$ as in the proof of Lemma \ref{lem.roughestimates}, that is, satisfying $\theta_\ell^2 + \theta_r^2 =1$ and such that $\theta_\star$ equals one in a neighborhood of the well support $D(x_\star, a)$. Using the IMS localization formula \cite[Prop.~4.2]{Raymond17} we that get
\[ \langle \mathscr{L} \Psi, \Psi \rangle = \langle \mathscr{L} \theta_\ell \Psi, \theta_\ell \Psi \rangle + \langle \mathscr{L} \theta_r \Psi, \theta_r \Psi \rangle - h^2 \| (|\nabla \theta_\ell |^2 + |\nabla \theta_r|^2 ) \Psi \|^2,\]
and by Agmon estimates, $\Psi$ is exponentially small on the support of $\nabla \theta_\star$, so that
\[ \langle \mathscr{L} \Psi, \Psi \rangle = \langle \mathscr{L} \theta_\ell \Psi, \theta_\ell \Psi \rangle + \langle \mathscr{L} \theta_r \Psi, \theta_r \Psi \rangle + \mathscr{O}(h^\infty)\| \Psi \|^2.\]
We insert the phase factor reflecting the gauge change and use \eqref{eq.gaugechange} to obtain
\[ \langle \mathscr{L} \Psi, \Psi \rangle = \langle \mathscr{L}_\ell e^{-i \frac{\sigma_\ell}{h}} \theta_\ell \Psi, e^{-i\frac{\sigma_\ell}{h}}\theta_\ell \Psi \rangle + \langle \mathscr{L}_r e^{-i \frac{\sigma_r}{h}} \theta_r \Psi, e^{-i \frac{\sigma_r}{h}}\theta_r \Psi \rangle + \mathscr{O}(h^\infty)\| \Psi \|^2.\]
Using next the spectral projections $\Pi_\star = |\Psi_\star \rangle \langle \Psi_\star |$ we find
\begin{multline*} \langle \mathscr{L} \Psi, \Psi \rangle \geq \mu_\ell | \langle e^{-i \frac{\sigma_\ell}{h}} \theta_\ell \Psi, \Psi_\ell \rangle |^2 + \mu_r| \langle e^{-i \frac{\sigma_r}{h}} \theta_r \Psi, \Psi_r \rangle |^2 \\ + \mu_{\ell,2} \|(I-\Pi_\ell) e^{-i \frac{\sigma_\ell}{h}} \theta_\ell \Psi \|^2 + \mu_{r,2} \| (I-\Pi_r)e^{-i \frac{\sigma_r}{h}} \theta_r \Psi \|^2+ \mathscr{O}(h^\infty)\| \Psi \|^2.
\end{multline*}
Note also that using the decay of $\psi_\star$, one is able to replace $\theta_\ell$ by the the cut-off function $\chi_\ell$ of \eqref{def.psil} up to exponentially small errors, so that
\[  \langle e^{-i \frac{\sigma_\ell}{h}} \theta_\ell \Psi, \Psi_\ell \rangle = \langle \Psi, e^{i\frac{\sigma_\ell}{h}} \chi_\ell \Psi_\ell \rangle + \mathscr{O} (h^\infty) \| \Psi \| = \mathscr{O}(h^\infty),\]
where in the last equality we have used the orthogonality assumption \eqref{eq.orthogonal}. In particular, $\| \Pi_\ell  e^{-i \frac{\sigma_\ell}{h}} \theta_\ell \Psi \| = \mathscr{O}(h^\infty) \| \Psi \|$, and we have
\[ \langle \mathscr{L} \Psi, \Psi \rangle \geq \big(\min( \mu_{\ell,2}, \mu_{r,2} ) + \mathscr{O}(h^\infty) \big) \| \Psi \|^2.\]
However, vector $\Psi$ is by assumption in the range of $\Pi$, meaning that
\[ \langle \mathscr{L} \Psi, \Psi \rangle \leq \min \big( \frac{\mu_{\ell}+\mu_{\ell,2}}{2}, \frac{\mu_r + \mu_{r,2}}{2} \big) \| \Psi \|^2.\]
Since $\mu_\ell - \mu_r = o(h^2)$ by assumption and $\mu_{\star,2}-\mu_{\star} \geq c h^2$ by Theorem~\ref{thm.single.magnetic.well}, the last two inequalities are incompatible unless $\Psi=0$. Hence ${\rm{Ran}} \, \Pi \cap \lbrace g_\ell, g_r \rbrace ^\perp = \lbrace 0 \rbrace$, which means that $(g_\ell,g_r)$ is a basis of the subspace ${\rm{Ran}}\, \Pi$.
\end{proof}

\subsection{Interaction matrix}

\begin{proposition}
Assume that $\mu_\ell - \mu_r = o(h^2)$. Then there exists an orthonormal basis $(e_\ell,e_r)$ of ${\rm{Ran}} \, \Pi$ such that
\[ e_\ell = \Psi_\ell + \mathscr{O}\big(h^\infty + |\langle \psi_r, \psi_\ell \rangle |\big), \qquad e_r = \Psi_r + \mathscr{O}\big(h^\infty + |\langle \psi_r, \psi_\ell \rangle |\big),\]
and the restriction of $\mathscr{L}$ to ${\rm{Ran}} \, \Pi$ is in this basis represented by the matrix
\[ M =
\begin{pmatrix}
\mu_\ell & w \\
\overline{w} & \mu_r
\end{pmatrix}
+ \mathscr{O}\big(|\langle \psi_\ell, \psi_r \rangle |^2 + |w|^2 +   e^{- 2 d_\ell(L-a-4\eta)/h} + e^{-2d_r(L-a-4\eta)/h}\big) \]
with $w:= \langle \big( \mathscr{L} - \frac{\mu_\ell + \mu_r}{2} \big) \psi_\ell, \psi_r \rangle$.
\end{proposition}
\begin{proof}
The action of $\mathscr{L}$ in the subspace ${\rm{Ran}} \, \Pi$ spanned by $(g_\ell,g_r)$ is expressed by the matrix
\[ L = \begin{pmatrix}
\langle \mathscr{L} g_\ell, g_\ell \rangle & \langle \mathscr{L} g_\ell, g_r \rangle \\
\langle \mathscr{L} g_r, g_\ell \rangle  & \langle \mathscr{L} g_r, g_r \rangle
\end{pmatrix}.
\]
By virtue of Lemma \ref{lem.g-psi}, we can replace in this expresssion $g_\star$ by $\psi_\star$,
\[
L= \begin{pmatrix}
\langle \mathscr{L} \psi_\ell, \psi_\ell \rangle & \langle \mathscr{L} \psi_\ell, \psi_r \rangle \\
\langle \mathscr{L} \psi_r, \psi_{\ell} \rangle & \langle \mathscr{L} \psi_r, \psi_r \rangle
\end{pmatrix}
+ \mathscr{O} \big( \mathcal E \big)
\]
with the error of order $\mathcal E =  e^{- 2 d_\ell(L-a-4\eta)/h} + e^{-2d_r(L-a-4\eta)/h}$. By the same lemma we can replace $\langle \mathscr{L} \psi_\star, \psi_\star \rangle$ by $\mu_\star$ with the same error, thus obtaining
\[ L = \begin{pmatrix}
\mu_\ell & w \\
\overline{w} & \mu_r
\end{pmatrix} + \frac{\mu_\ell + \mu_r}{2}T + \mathscr{O}\big( \mathcal E \big) \]
with $w= \langle (\mathscr{L} - \frac{\mu_\ell + \mu_r}{2}) \psi_\ell, \psi_r \rangle$ and
\[ T= \begin{pmatrix}
0 & \langle \psi_\ell, \psi_r \rangle \\
\langle \psi_r, \psi_\ell \rangle & 0
\end{pmatrix}.\]
However, the basis $(g_\ell,g_r)$ is not orthonormal. For this reason we consider the Gram matrix
\[ G = \begin{pmatrix}
\langle g_\ell, g_\ell \rangle & \langle g_\ell, g_r \rangle \\
\langle g_r, g_\ell \rangle & \langle g_r, g_r \rangle
\end{pmatrix},
\]
the leading term of which can be by Lemma \ref{lem.g-psi} expressed as
\[
G = I + T + \mathscr{O}(\mathcal E).
\]
Then the basis ${e_\ell \choose e_r} = G^{-\frac 12} {g_\ell \choose g_r}$ is orthonormal, and the action of $\mathscr{L}$ in it is expressed by the matrix
\[ M = G^{-\frac 12} L G^{- \frac 12} =
\begin{pmatrix}
\mu_\ell & w \\ \overline{w} & \mu_r
\end{pmatrix}
+\mathscr{O}(|T|^2 + \mathcal{E} + |w|^2),
 \]
where we used Taylor expansion of the square root, $G^{-\frac 12}= I- \frac 12 T + \mathscr{O}(T^2)$.
\end{proof}

 In order to prove Theorem \ref{thm.doublewell}, it remains to estimate $w$ and the error terms $\langle \psi_\ell, \psi_r \rangle$ and $e^{-2d_\star(L-a-4\eta)/h}$. This is the purpose of the next section.

\subsection{Estimation of the interaction coefficient and the error terms}

The sought estimates on $w$, $\langle \psi_\ell, \psi_r \rangle$ and of the error term $e^{-2d_\star(L-a-4\eta)/h}$ are given in Lemmata \ref{lem.w}, \ref{lem.scalarproduct} and \ref{lem.error}, respectively. They are technical, and the two first ones are based on the Laplace method. We again follow the estimate ideas used in \cite{FMR}.
\begin{lemma}\label{lem.w}
There are a nonzero $c_0$ and $\nu \in \mathbb R$ such that
\[ w = c_0 h^\nu e^{-S/h} \big(1+o(1)\big) + \mathscr{O} \big( (\mu_\ell - \mu_r) |\langle \psi_\ell, \psi_r \rangle | \big) \]
holds as $h \to 0$, where $S$ is given by \eqref{defS}.
\end{lemma}
\begin{proof} The proof is essentially the same as in \cite[Section 5]{FMR}, thus we recall here only the main steps. In $w$ we replace $\mu_r$ by $\mu_\ell$ and we find
\[w = \tilde w + \mathscr{O}((\mu_\ell - \mu_r) \langle \psi_\ell, \psi_r \rangle), \;\; \text{where}\;\; \tilde w := \langle (\mathscr{L}-\mu_\ell) \psi_\ell, \psi_r \rangle. \]
Recalling the definition \eqref{def.psil} of $\psi_\ell$ and using \eqref{eq.gaugechange} we get
\[ \tilde{w} = \langle [ \mathscr{L}, \chi_\ell] e^{i \frac{\sigma_\ell}{h}} \Psi_\ell, e^{i \frac{\sigma_r}{h}} \Psi_r \rangle .\]
Since $\mathscr{L}=P^2$ with $P=ih\nabla + A$, we can rewrite the commutator expression as
\begin{align*}
 \tilde{w} &= \langle (P \cdot [P,\chi_\ell] + [P,\chi_\ell] \cdot P) e^{i \frac{\sigma_\ell}{h}}\Psi_\ell, e^{i \frac{\sigma_r}{h}} \Psi_r \rangle \\
 &= \langle [P_1, \chi_\ell] e^{i \frac{\sigma_\ell}{h}}\Psi_\ell, P_1 e^{i \frac{\sigma_r}{h}} \Psi_r \rangle + \langle [P_1,\chi_\ell] P_1 e^{i \frac{\sigma_\ell}{h}}\Psi_\ell, e^{i \frac{\sigma_r}{h}} \Psi_r \rangle\\
 &= ih \int_{x_1>0} \chi_\ell'(x_1) e^{\frac ih (\sigma_\ell - \sigma_r) }\big( \Psi_\ell \overline{(ih\partial_1 + A_{r,1}) \Psi_r}  + ((ih\partial_1 + A_{\ell,1}) \Psi_\ell) \overline{\Psi_r}\big) \dd x.
\end{align*}
where we have used the fact that $\chi_\ell$ is independent of $x_2$ and  $\chi_\ell'$ is supported in the half-plane $x_1>0$, together with $e^{-i \frac{\sigma_\star}{h}} P e^{i \frac{\sigma_\star}{h}} = P_\star$. Integrating then the obtained expression by parts, we get
\begin{align*}
\tilde{w} &= - ih\int_{\lbrace x:\,x_1 = 0 \rbrace } e^{\frac{i}{h}(\sigma_\ell - \sigma_r)} \big( \Psi_\ell \overline{(ih\partial_1 + A_{r,1}) \Psi_r}  + ((ih\partial_1 + A_{\ell,1}) \Psi_\ell) \overline{\Psi_r}\big) \dd x_2\\
&\quad + \int_{x_1>0} \chi_\ell(x_1) e^{\frac{i}{h}(\sigma_\ell - \sigma_r)}  \big( \Psi_\ell \overline{(ih\partial_1 + A_{r,1})^2 \Psi_r}  - ((ih\partial_1 + A_{\ell,1})^2 \Psi_\ell) \overline{\Psi_r}\big) \dd x\\
&= -ih \int_{\{x:\,x_1=0\}} e^{\frac{i}{h}(\sigma_\ell - \sigma_r)}  \big( \Psi_\ell \overline{(ih\partial_1 + A_{r,1}) \Psi_r}  + ((ih\partial_1 + A_{\ell,1}) \Psi_\ell)  \overline{\Psi_r}\big) \dd x_2 \\
&\quad + \int_{x_1>0} \chi_\ell(x_1) e^{\frac{i}{h}(\sigma_\ell - \sigma_r)}  \big( \Psi_\ell \overline{\mathscr{L}_r \Psi_r}  -\mathscr{L}_\ell \Psi_\ell \overline{\Psi_r}\big) \dd x,
\end{align*}
where the last equality is justified using integration by parts with respect to $x_2$. Using further the eigenvalue equation for $\Psi_\ell$ and $\Psi_r$, we infer  that
\begin{multline}
w= -ih \int_{\{x:\,x_1=0\}} e^{\frac{i}{h}(\sigma_\ell - \sigma_r)}  \big( \Psi_\ell \overline{(ih\partial_1 + A_{r,1}) \Psi_r}  + ((ih\partial_1 + A_{\ell,1}) \Psi_\ell) \overline{\Psi_r}\big) \dd x_2 \\
+ \mathscr{O}((\mu_r - \mu_\ell) \langle \psi_\ell, \psi_r \rangle).
\end{multline}
To estimate the above integral, we employ the Laplace method. In view of the complex phase involved, we have first to extend the function to complex values of the argument, namely
 \[ x_2 \mapsto x_2 + iq, \;\; {\rm where} \;\; q = - \sqrt{ \frac{L^2}{4} - \frac{M_\ell + M_r}{b_1}}\]
with $M_\star := \frac{1}{2\pi} \int_{\R^2} \big( b_1 - B_\star \big) \dd x.$ This can be done in view of the analyticity of $\Psi_\ell$ and $\Psi_r$ (see Lemma \ref{lem.psi.asymp} below and \cite{FMR} for further details). We obtain
\begin{align*}
\mathcal I &:=  -ih \int_{\{x:\,x_1=0\}} e^{\frac{i}{h}(\sigma_\ell - \sigma_r)}  \big( \Psi_\ell \overline{(ih\partial_1 + A_{r,1}) \Psi_r}  + ((ih\partial_1 + A_{\ell,1}) \Psi_\ell) \overline{\Psi_r}\big) \dd x_2 \\
&= -ih \int_{\mathbb R} e^{\frac{i}{h}(\sigma_\ell - \sigma_r)(0,x_2+iq)}  \big( \Psi_\ell \overline{(ih\partial_1 + A_{r,1}) \Psi_r}  + ((ih\partial_1 + A_{\ell,1}) \Psi_\ell) \overline{\Psi_r}\big)(0,x_2+iq) \dd x_2.
\end{align*}
Next we use Lemma \ref{lem.psi.asymp} to replace $\Psi_\ell$ and $\Psi_r$ by the leading terms of their asymptotic expansions with the phase given by  \eqref{eq.sigma.diff} and \eqref{eq.phase.g}. This gives
\begin{align*}
\mathcal I = c_\ell c_r\, e^{- \frac{d_\ell(L/2) + d_r(L/2)}{h}} \int_{\mathbb R} e^{- \frac{g(x_2 + iq)}{h}} \omega(x_2 + i q) \,\dd x_2 \,\big( 1+ o(1) \big),
\end{align*}
with $g(y) = \frac{b_1}{2} y^2 + i \frac{b_1 L}{2} y - (M_\ell + M_r) \ln \big( 1 + i \frac{2 y}{L} \big) $ and
\begin{align*}
 \omega(y) &= \big( y+i \textstyle{\frac L2} \big) \left( \frac{M_\ell + M_r}{ \frac{L^2}{4}+y^2} - b_1 \right) \left( \frac{\frac{b_1L^2}{8} - M_\ell}{\frac{b_1 L^2}{8} -M_\ell + \frac{b_1 y^2}{2}} \right)^{\delta_\ell} \left( \frac{\frac{b_1L^2}{8} - M_r}{\frac{b_1L^2}{8} -M_r + \frac{b_1 y^2}{2}} \right)^{\delta_r} \\
 &= g'(y)  \left( \frac{\frac{b_1L^2}{8} - M_\ell}{\frac{b_1 L^2}{8} -M_\ell + \frac{b_1 y^2}{2}} \right)^{\delta_\ell} \left( \frac{\frac{b_1L^2}{8} - M_r}{\frac{b_1L^2}{8} -M_r + \frac{b_1 y^2}{2}} \right)^{\delta_r}
 \end{align*}
with some $\delta_\star>0$. The function $x_2 \mapsto {\rm Re}\, g(x_2+iq)$ has a unique minimum at $x_2=0$, and it is non-degenerate. For this reason, we are allowed to use the Laplace method by which there are $c_0 \neq 0$ and $\nu \in \mathbb R$ such that
\[ \mathcal I = c_0\, h^\nu e^{- \frac{d_\ell(L/2) + d_r(L/2)+ g(iq)}{h}} (1+ o(1))\]
holds as $h\to 0$, and since
\begin{multline*}
g(iq) =  \frac{M_\ell + M_r}{2} - \frac{b_1L^2}{8} + \frac{b_1L}{2} \sqrt{\frac{L^2}{4}-\frac{M_\ell + M_r}{b_1}} \\ - (M_\ell + M_r) \ln \left( 1 + \sqrt{ 1 - \frac{4(M_\ell + M_r)}{b_1 L^2}}\right),
\end{multline*}
this concludes the proof.
\end{proof}
\begin{lemma}\label{lem.scalarproduct}
$\langle \psi_\ell, \psi_r \rangle = \mathscr{O}(h^{-1}|w|)$ holds as $h \to 0$.
\end{lemma}
\begin{proof}
By definition of $\psi_\ell$ and $\psi_r$ we have
\[ \langle \psi_\ell, \psi_r \rangle = \int_{\R^2} \chi_\ell \chi_r\, e^{\frac{i}{h}(\sigma_\ell - \sigma_r)} \Psi_\ell \overline{\Psi}_r \,\dd x.\]
This integral can be estimated in a way similar to that of Lemma \ref{lem.w}. We recall that $\chi_\ell$ and $\chi_r$ depend on $x_1$ only. On the support of $\chi_\ell \chi_r$, the functions $\Psi_\ell$ and $\Psi_r$ are analytic with respect to $x_2$. This makes it possible to extend them to complex values of the argument, $x_2 \mapsto x_2 +iq(x_1)$, where
\[ q(x_1) = - \sqrt{\left( \frac L2 - |x_1| \right)^2 - \frac{M_\ell + M_r}{b_1}} \]
so that $x_2$ stays in the analyticity domain of $\Psi_\ell$ and $\Psi_r$ and we get
\[ \langle \psi_\ell , \psi_r \rangle = \int_{\R^2} \left( \chi_\ell \chi_r\, e^{\frac{i}{h}(\sigma_\ell - \sigma_r)} \Psi_\ell \overline{\Psi}_r \right) (x_1, x_2+iq(x_1)) \,\dd x.\]
Now we use the asymptotic formula from Lemma \ref{lem.psi.asymp} which yields
\[ \langle \psi_\ell, \psi_r \rangle = c_\ell c_r h^{-1} e^{- \frac{d_\ell(L/2) + d_r(L/2)}{h} } \int_{\R^2} e^{- \frac{f(x_1,x_2+iq_1)}{h}} \tilde \omega (x_1,x_2+iq(x_1)) \,\dd x, \]
with
\begin{multline}
f(x) = \frac{b_1}{4} \left( |x-x_\ell|^2 + |x-x_r|^2 - \frac{L^2}{2}  +2iLx_2 \right) \\- M_\ell \ln \left( \frac 2L \left( \frac L2 +x_1 +ix_2 \right) \right) - M_r \ln \left( \frac 2L \left( \frac L2 -x_1 +ix_2 \right) \right)
\end{multline}
and
\begin{equation}
\tilde \omega(x_1,x_2) =  \left( \frac{\frac{b_1L^2}{8} - M_\ell}{\frac{b_1}{2}\left( \frac L2 + x_1 \right)^2 +\frac{b_1x_2^2}{2}   -M_\ell}  \right)^{\delta_\ell}
 \left( \frac{\frac{b_1L^2}{8} - M_r}{\frac{b_1}{2}\left( \frac L2 - x_1 \right)^2 +\frac{b_1x_2^2}{2}   -M_r}  \right)^{\delta_r}.
\end{equation}
The function $x_1 \mapsto {\rm Re} f(x_1,x_2+iq(x_1))$ has a non-degenerate global minimum at $x_1=0$, and consequently, the Laplace method gives
\begin{equation}
| \langle \psi_\ell, \psi_r \rangle | \leq C h^{- \frac 12} e^{- \frac{d_\ell(L/2) + d_r(L/2)}{h} } \int_{\R} e^{- {\rm{Re}} \frac{f(0,x_2+iq(0))}{h}} \tilde \omega(0,x_2+iq(0)) \,\dd x_2.
\end{equation}
Note further that $f(0,x_2) = g(x_2)$ and $\tilde \omega(x_2) = \frac{1}{g'(x_2)}\omega(x_2)$ holds with the notations of the proof of Lemma \ref{lem.w}, thus we get
\begin{equation}
| \langle \psi_\ell, \psi_r \rangle | \leq  C h^{-1}|w|,
\end{equation}
which concludes the proof.
\end{proof}

Finally, we need an estimate from which the condition on well spacing $L$ follows:
\begin{lemma}\label{lem.error}
If $L > (2+\sqrt 6) a$ then for $\star \in \lbrace \ell, r \rbrace$ we have
\[ 2d_\star(L-a) > S .\]
In particular, if $\eta >0$ is small enough then $e^{-2d_\star(L-a-4\eta)} = o(h^\nu e^{-S/h})$ holds for all $\nu \in \mathbb R$.
\end{lemma}
\begin{proof}
The proof is identical with that of \cite[Lemma 5.2]{FMR}.
\end{proof}

\noindent Putting the above results together we complete the proof of Theorem \ref{thm.doublewell}.

\appendix

\section{About the single well ground state}

\begin{lemma}\label{lem.psi.formule}
In the exterior of the well, $|x-x_\star| \geq a$, we have
\[ \Psi_\star(x) = C_\star |x-x_\star|^{\gamma_{\star}} \int_0^{\infty} e^{- \frac{b_1}{4h} (1+2t) |x-x_\star|^2} m_\star(t) \,\dd t\]
with
\[ \gamma_\star = \frac{M_\star}{h}, \quad \delta_\star =\frac{b_1 h - \mu_\star}{2hb_1}, \quad m_\star(t) = t^{\delta_\star -1} (1+t)^{\gamma_\star - \delta_\star},\]
where $M_\star = \frac{1}{2\pi} \int_{\R^2} \big( b_1 - B_\star(x) \big) \dd x$ and
\[ C_\star = \frac{\Psi_\star(0)}{|x_\star|^{\gamma_\star}} \Big( \int_0^\infty e^{- \frac{b_1}{4h}(1+2t)|x_\star|^2} m_\star(t) \dd t \Big)^{-1}.\]
\end{lemma}
\begin{proof}
The claim is established using the Kummer functions, see \cite{FMR}.
\end{proof}

\begin{lemma}\label{lem.psi.asymp}
For $x_1 \in \big(- \frac L2 +a, \frac L2 -a \big)$, the function $x_2 \mapsto \Psi_\star(x_1,x_2)$ has a holomorphic extension in the second argument to the region
\[ \Omega(x_1) = \big\lbrace x_2 \in  \mathbb C , \, \big(x_1 \pm \frac L2\big)^2 + ({\rm{Re}} x_2)^2 - ({\rm{Im}} x_2)^2 > a^2 \big\rbrace. \]
Moreover, there exists a positive constant $c_\star >0$ such that
\[ \Psi_\star(x) = c_\star h^{- \frac 12} \left( \frac{\frac{b_1L^2}{8} - M_\star}{\frac{b_1}{2}|x-x_\star|^2 -M_\star}  \right)^\delta  \frac{|x-x_\star|^{\gamma_\star}}{|x_\star|^{\gamma_\star}} e^{- \frac{d_\star(L/2)}{h}} e^{- \frac{b_1}{4h}(|x-x_\star|^2 - |x_\star|^2)} \big( 1+ \mathscr{O}(h) \big) \]
holds locally uniformly for $x_1 \in \big(- \frac L2 +a, \frac L2 -a \big)$ and $x_2 \in \Omega(x_1)$.
\end{lemma}
\begin{proof}
The claim follows from Lemma \ref{lem.psi.formule}. Since $M_\star < \frac{b_1 a^2}{2}$, the integral defining $\Psi_\star$ in holomorphic on $\Omega(x_1)$. One can estimate $\Psi_\star$ and $C_\star$ using the Laplace method; the value of $\Psi_\star(0)$ is given by the WKB approximation, Proposition~\ref{prop.WKB}.
\end{proof}

\section{About the complex phase}

In the above reasoning we needed the phase shift $\sigma_r - \sigma_\ell$ defined by \eqref{def.sigma} expressing the gauge difference between the two wells. It can be calculated as follows,
\begin{align*}
\sigma_r(x) - \sigma_\ell(x) &= \int_{[0,x]} (A_\ell - A_r) \cdot \dd s \\
&=\int_0^{x_1} (A_\ell - A_r)_{1}(u,0) \,\dd u + \int_0^{x_2} (A_\ell - A_r)_2 (x_1,u) \,\dd u \\
&= \int_0^{x_2} (A_\ell - A_r)_2(x_1,u) \,\dd u,
\end{align*}
with
\begin{align*}
A_\star(x) = \int_0^1 B_\star(x_\star + t(x-x_\star))t \,\dd t
\begin{pmatrix}
-x_2 \\ x_1 \pm \frac L2
\end{pmatrix}
= \left( \frac{b_1}{2} - \frac{M_\star}{|x-x_\star|^2} \right) \begin{pmatrix}
-x_2 \\ x_1 \pm \frac L2
\end{pmatrix}
\end{align*}
and the sign $+$ corresponding $A_\ell$ and $-$ respectively to $A_r$. In particular, we have
\[ A_{\star,2}(x) = \frac{b_1}{2} \big(x_1 \pm \frac L2\big) - \frac{M_\star (x_1 \pm \frac L2)}{(x_1 \pm \frac L2)^2 + x_2^2} \]
and
\[ A_{\ell,2}(x) - A_{r,2}(x) = \frac{b_1 L}{2} +\frac{M_r (x_1 - \frac L2)}{(x_1 - \frac L2)^2 + x_2^2} - \frac{M_\ell (x_1 + \frac L2)}{(x_1 + \frac L2)^2 + x_2^2}.\]
Consequently, the phase difference equals
\begin{equation}\label{eq.sigma.diff}
 \sigma_r(x) - \sigma_\ell(x) = \frac{b_1 L}{2}x_2 - M_r \arctan \left( \frac{x_2}{\frac L2 - x_1} \right) - M_\ell \arctan \left( \frac{x_2}{x_1 + \frac L2} \right).
\end{equation}
In Lemma~\ref{lem.scalarproduct} we have also encounter the phase
\begin{multline}\label{eq.phase}
f(x) = \frac{b_1}{4} \left( |x-x_\ell|^2 + |x-x_r|^2 - \frac{L^2}{2} \right) - M_r \ln \left( \frac{|x-x_r|}{|x_r|} \right) - M_\ell  \ln \left( \frac{|x-x_\ell|}{|x_\ell|} \right) \\
+ i \frac{b_1 L}{2} x_2 - i M_r \arctan \left( \frac{x_2}{\frac L2 - x_1} \right) -i M_\ell \arctan \left( \frac{x_2}{x_1 + \frac L2} \right).
\end{multline}
which can be rewritten using the complex logarithm
\begin{multline}
f(x) = \frac{b_1}{4} \left( |x-x_\ell|^2 + |x-x_r|^2 - \frac{L^2}{2}  +2iLx_2 \right) \\- M_\ell \ln \left( \frac 2L \left( \frac L2 +x_1 +ix_2 \right) \right) - M_r \ln \left( \frac 2L \left( \frac L2 -x_1 +ix_2 \right) \right).
\end{multline}
Note that the function $g(y) = f(0,y)$ of Lemma~\ref{lem.w} is given by
\begin{equation}\label{eq.phase.g}
g(y) = \frac{b_1y^2}{4} + i \frac{b_1 Ly }{2} - (M_\ell + M_r) \ln \left( 1 + i \frac{2y}{L} \right).
\end{equation}

\section*{Acknowledgments}
The work of P.E. was in part supported by the European Union's Horizon 2020 research and innovation programme under the Marie Sk\l odowska-Curie grant agreement No 873071.

The work of L.M. is funded by the European Union (via the ERC Advanced Grant MathBEC - 101095820). Views and opinions expressed are however those of the author only and do not necessarily reflect those of the European Union or the European Research Council. Neither the European Union nor the granting authority can be held responsible for them.
	
\bibliographystyle{plain}
\bibliography{UTF-8biblio}

\end{document}